\documentclass[11pt]{article}

\usepackage[margin = 1in]{geometry}
\usepackage{times}

\usepackage{amssymb,amsmath,amsthm,hyperref,psfrag,epsfig}
\usepackage{algorithmic}
\usepackage{algorithm}
\usepackage{color}
\hypersetup{colorlinks=true}

\newtheorem{theorem}{Theorem}
\newtheorem{lemma}[theorem]{Lemma}
\newtheorem{corollary}[theorem]{Corollary}
\newtheorem{proposition}[theorem]{Proposition}

\newcommand{\floor}[1]{\lfloor #1 \rfloor}
\newcommand{\ceil}[1]{\lceil #1 \rceil}
\newcommand{\ang}[1]{\langle #1 \rangle}
\newcommand{\on}{\textit{on}}
\newcommand{\off}{\textit{off}}

\newcommand{\F}{\mathcal F}
\newcommand{\E}{\text{E}}

\usepackage{color}

\newcommand{\LOOKr}[1]{{\textcolor{red}{#1}}}

\begin{document}

\title{Near Optimality in Covering and Packing Games\\ by Exposing Global Information}
\author{
Maria-Florina Balcan\thanks{School of Computer Science, Georgia
Institute of Technology, Atlanta GA 30332. Email:
ninamf@cc.gatech.edu.}
\and Sara Krehbiel\thanks{School of Computer Science, Georgia
Institute of Technology, Atlanta GA 30332. Email: sarak@gatech.edu.}
\and Georgios Piliouras\thanks{School of Electrical \& Computer
Engineering, Georgia Institute of Technology, Atlanta GA 30332.
Email: georgios.piliouras@ece.gatech.edu}
 \and
 Jinwoo Shin\thanks{Algorithms \& Randomness Center, Georgia
Institute of Technology, Atlanta GA 30332. Email:
jshin72@cc.gatech.edu.}
}

\pagestyle{empty}

\maketitle

\begin{abstract}
Covering and packing problems can be modeled as games to encapsulate interesting social and engineering settings.  These games have a high Price of Anarchy in their natural formulation.  However, existing research applicable to specific instances of these games has only been able to prove fast convergence to arbitrary equilibria. This paper studies general classes of covering and packing games with learning dynamics models that incorporate a central authority who broadcasts weak, socially beneficial signals to agents that otherwise only use local information in their decision-making.  Rather than illustrating convergence to an arbitrary equilibrium that may have very high social cost, we show that these systems quickly achieve near-optimal performance.

In particular, we show that in the public service advertising model of \cite{BBM}, reaching a small constant fraction of the agents is enough to bring the system to a state within a $\log n$ factor of optimal in a broad class of set cover and set packing games or a constant factor of optimal in the special cases of vertex cover and maximum independent set, circumventing social inefficiency of bad local equilibria that could arise without a central authority.  We extend these results to the learn-then-decide model of \cite{bbm-ics-10}, in which agents use any of a broad class of learning algorithms to decide in a given round whether to behave according to locally optimal behavior or the behavior prescribed by the broadcast signal.  The new techniques we use for analyzing these games could be of broader interest for analyzing more general classic optimization problems in a distributed fashion.
\end{abstract}


\pagestyle{plain}
\section{Introduction}
Set covering and packing problems are important and interesting not only from a classical optimization point of view, but also as a game theoretic framework for analyzing social problems in which willful agents are inherent cost minimizers and for solving engineering systems problems in which programmable agents have some degree of autonomy in seeking solutions to distributed optimization problems.  In this paper, we model covering and packing problems as games, and we use models from learning theory to describe local decision making by players in these games.  As opposed to previous work, we are interested in demonstrating convergence not to arbitrary local equilibria but to states that are low cost relative to the global optimum.  We accomplish this by incorporating a globally-informed central authority into natural behavior dynamics.


\paragraph{Problem.} Given a universe of elements with associated costs and a collection of sets of these elements, the minimum weighted set cover optimization problem is to choose the lowest cost subset of elements such that each set is represented by at least one chosen element.  While this problem is NP-hard, good approximation algorithms exist. However, such algorithms tend to be  centralized in nature and require global knowledge.

\paragraph{Game.} We analyze a setting in which a central authority knows a good approximation, but elements are modeled as only locally aware agents with cost functions representing a natural distributed game interpretation of the core optimization problem.  We generalize the problem by not requiring total coverage, rather the importance of covering a given set is determined by its set weight.  Each element $i$ that chooses to be \on\ incurs his own cost $c_i$, and each element $i$ that is \off\ pays the sum of the weights of sets he participates in that do not contain any other \on\ element.  If the element costs are all smaller than the set weights, then the cost-minimizing set of \on\ elements  is also the optimal set cover.  If additionally each set is of size two, then this is the special case of a minimum weighted vertex cover problem.  By simply redefining the cost structure so that $i$ pays $c_i$ if he is \off\ and the sum of weights of fully-covered sets he participates in if he is \on, we can interpret this new game as a packing problem with maximum independent set as a special case.

\paragraph{Social and engineering applications.} Our motivation for this game theoretic approach is two-fold.  The first setting is a social one in which agents have inherent costs associated with being \on\ or \off\ that correlate with the social objective.  As a concrete example, suppose government wishes to set up a network of offices, say homeless shelters, that perform some service to the local community.  Society would like the lowest cost solution that adequately addresses the needs of most communities, but for political reasons it may not be possible to enforce an optimal solution in a top-down manner.  Furthermore, individual counties have competing interests in that they desire their own area to be served but incur some cost by opening a shelter.

Another motivation is the setting in which non-autonomous agents are programmed to make decisions based on their surroundings.  The extensive literature on cooperative control has shown that in this setting many optimization problems can be conveniently solved in a distributed fashion by endowing agents with artificial individual objective functions and cost-minimizing behavior.  Many of these games and dynamics models result in convergence to a Nash equilibrium, or local optimum.  In particular, several papers have modeled sensor networks as a special case of our set cover game.  The elements are autonomous sensors, and a geographic region is a set consisting of elements corresponding to sensors that could cover that region.  A sensor that is \on\ is charged some fixed cost, whereas a sensor that is \off\ is charged a cost proportional to the number or importance of its adjacent regions that are uncovered by any other sensor.  This application is particularly well-suited for cooperative control because sensors can only observe the behavior of other sensors in their neighborhoods, and the structure of the network may not be known ahead of time, making it impossible for a central designer to program the sensors with an optimal solution.

\paragraph{Equilibrium quality and dynamics models.} Much of the work on cooperative control and
dynamics-based algorithmic game theory only guarantees that systems
converge to some equilibrium. Many games, however, have a high Price
of Anarchy (PoA), where PoA means the worst case ratio between the
social cost in an equilibrium and that of the global optimal
configuration (see Section \ref{general} for its formal definition).
The following special case illustrates that PoA is $\Omega(n)$ in
our set cover game. Suppose $n$ agents (or players) are charged some
amount $c<1$ when they are \on\ and otherwise penalized 1 for every
incident uncovered set. Then a star graph in which vertices are
agents and edges are sets has a global optimum with only the center
\on, yielding social cost $c$, compared to a low quality Nash
Equilibrium in which only the center is \off, yielding social cost
$c(n-1)$.

The more general problem of dynamics for games with high PoA is
addressed in \cite{BBM, bbm-ics-10}, in which authors propose three
models of distributed and semi-selfish social behavior in a general
repeated game setting. The models share the common feature that a
central authority has knowledge of some joint strategy profile with
low social cost, and this authority broadcasts this strategy in the
hopes that players will adopt their prescribed strategies.
Specifically, the public service advertising model (PSA) of
\cite{BBM} assumes that each agent independently has an $\alpha$
probability of receiving and temporarily adopting the advertising
strategy.  Those that do not receive and adopt their prescribed
strategy behave in a myopic best response manner.  This model is
well-suited for an engineering systems setting, where we do not
expect all components to receive the central authority's signal. The
learning models of \cite{bbm-ics-10} assume that each agent uses any
of a broad class of learning algorithms to continually choose
between acting according to their local best response move and their
broadcasted signal. In the learn-then-decide (LTD) model, agents
eventually commit to one of these options. These models are better
motivated by a social setting where agents that are only locally
aware are interested in exploring the advertising strategy with the
hopes that it will benefit them personally.  These papers provide
high quality guarantees for particular games, including fair
cost-sharing and party affiliation games.

\paragraph{Our results.} The positive theoretical guarantees about social welfare in the outcomes of the games
studied in the advertising and learning models of \cite{BBM,
bbm-ics-10} serve as motivation to use these models in studying our
general general classes of set cover and packing games, which apply
to engineering systems applications such as sensor networks as well
as more purely game theoretic settings.  For the case where costs of
agents and weights of sets are bounded below and above by constants,
we show the following:


\begin{itemize}
\item[R1.] In vertex cover games\footnote{As mentioned earlier, a set
cover game where each set has  size  $2$ is called a {\em vertex
cover game}, and in such games equilibria have natural connections
to vertex covers in the graph induced by the sets (i.e.\ edges).},
we show that for any advertising strategy $s^{ad}$, $$\mbox{the
dynamics of agents converges to a state of expected cost
$O(\texttt{cost}(s^{ad}))$ in PSA and LTD models.}$$

\item[R2.] In set cover games, we show that for any advertising strategy
$s^{ad}$, $$\mbox{the dynamics of agents converges to a state of
expected cost $\begin{cases}
O(\Delta_2\cdot\texttt{cost}(s^{ad})^2)&\mbox{in PSA model}\\
O(\Delta_2^2\cdot\texttt{cost}(s^{ad})^2)&\mbox{in LTD model}
\end{cases}
$,}$$where $\Delta_2$ is the maximum number of sets containing given
two agents.
\item[R3.] In set cover games, we show that for a specific advertising strategy
$s^{ad}$,
$$\mbox{the dynamics of agents converge to a state of cost $O(\texttt{cost}(s^{ad}))$ with high probability in PSA model.}$$
Moreover, we present a poly-time algorithm to find such a specific
$s^{ad}$ of low cost, i.e.\
$$\texttt{cost}(s^{ad})=O(\Delta_2\log n\cdot OPT),\qquad\mbox{where}~~OPT~\mbox{is the optimal (social) cost}.$$
\end{itemize}

Furthermore, we emphasize that all the above convergence guarantees
happen in polynomial number of steps in terms of the number of
agents. As we mentioned earlier, without such advertising
strategies, agents can be an inefficient equilibrium state of cost
$\Omega(n)\cdot OPT$, even restricted to vertex cover games (i.e.\
$\Delta_2=1$). We also discuss extensions to the case where the
costs of agents and weights of sets are not bounded below or above
by constants.

\paragraph{Related work.}
Achieving global coordination in distributed multi-agent systems is
a central problem of control theory with multiple real-world
applications (see \cite{Shamma} and references therein). More
specifically, several papers consider game theoretic formulations of
covering problems which are inspired by practical sensor network
problems \cite{modelpaper,diffusion,decentralizeddesign,cgl}. In particular, \cite{cgl} analyzes a game that is a specific case of the problem addressed in this paper.  However, \cite{cgl} and many other control theory papers guarantee only convergence to stable states which are
locally optimal. Since these games often have a high Price of
Anarchy \cite{PoA,book07}, the results do not translate to global
performance guarantees.

A number of approaches have been explored to
circumvent such bad PoA results. In
\cite{Sharma06stackelbergthresholds} the authors assume that the
authorities enjoy complete control over some fraction of the agents.
Similarly, \cite{Kempe03maximizingthe,Kempe05influentialnodes} focus on the problem of identifying and controlling the influential nodes of a network.  While we also use a special type of advertising for improved results in Theorem \ref{PSAoutcome-2}, we do not require particular control over certain agents.  Rather, the models we use from \cite{BBM, bbm-ics-10} incorporate strategic behavior for all agents.   Another line of
research offers stronger performance guarantees using specific learning algorithms
that employ equilibrium
selection \cite{paperstoc09,BookCL,BookGTL} or cyclic
behavior \cite{paperics11}.
Unfortunately, these techniques do not yield guarantees of fast convergence to good states in our class of games.

Our analysis builds on the works of \cite{BBM, bbm-ics-10}, in which
authors propose game theoretic models of distributed and
semi-selfish social behavior in a general repeated game setting. The
models share the common feature that a central authority has
knowledge of some joint strategy profile with low social cost, and
this authority broadcasts this strategy in the hopes that players
will adopt their prescribed strategies. These papers provide quality guarantees for particular
games, including fair cost-sharing and party affiliation games.
By using these models, we do not have to make the hard choice between enforcing top-down solutions
(which may be infeasible in both engineering systems and social settings) and poor performance guarantees.  Instead, we show that for a broad class of covering and packing problems,
incorporating mild influence from a weak central authority guides the system into a near-optimal state when
agents are only optimizing locally.

\section{Preliminaries} 
\label{setup}
\subsection{Background on General Games}\label{general}
We represent a general game as a triple $\mathcal
G=\ang{N,(S_i),(\texttt{cost}_i)}$, where $N$ is a set of $n$
players, $S_i$ is the finite action space of player $i\in N$, and
$\texttt{cost}_i$ denotes the cost function of player $i$. The joint
action space of the players is $S = S_1 \times \dots \times S_n$.
For a joint action $s \in S$, we denote by $s_{-i}$ the actions of
all players $j \neq i$. Players' cost functions map joint actions to
non-negative real numbers, i.e.\ $\texttt{cost}_i : S \to
\mathbb{R}^+$ for all $i\in N$. In this paper, we define a social
cost function, $\texttt{cost} : S \to \mathbb{R}$, simply as the
summation of individual players' costs. The optimal social cost is
denoted by $$OPT=\min_{s\in S} \texttt{cost}(s).$$

Given a joint action $s$, the {\em best response} of player $i$ is
the set of actions that minimizes player $i$'s cost subject to the
other players' fixed actions $s_{-i}$, i.e. $$BR_i(s_{-i}) =
{\arg\min}_{a\in S_i} \texttt{cost}_i(a, s_{-i}).$$ {\em Best
response dynamics} is a process in which at each time step, an
arbitrary player not already playing best response updates his
action to one in his current best response set. A joint action $s
\in S$ is a {\em pure Nash equilibrium} if no player $i \in N$ can
benefit from deviating to another action, namely, $s_i \in
BR_i(s_{-i})$ for every $i \in N$.

A game $\mathcal{G}$ is called an {\em exact potential game}
\cite{ms96} if there exists a potential function $\Phi:S\to
\mathbb{R}$ such that $$\texttt{cost}_i(a',s_{-i})
-\texttt{cost}_i(a,s_{-i})= \Phi(a',s_{-i})-\Phi(a,s_{-i}),$$ for
all $i\in N$, $s_{-i}\in S_{-i}$, and $a,a'\in S_i$. For general
potential games, only the signs of both sides of these equations
must be equal. While general games are not guaranteed to have a pure
Nash equilibrium, all finite potential games do and furthermore best
response dynamics in such games always converges to a pure Nash
equilibrium \cite{ms96,book07}. However, the convergence time can be
exponentially large in terms of the number of players in general.

Two well known concepts for quantifying the inefficiency of
equilibria relative to non-equilibria are  {\em Price of Anarchy}
and {\em Price of Stability}. For $\mathcal{N}(\mathcal{G})$ the set
of pure Nash equilibria of game $\mathcal{G}$, Price of Anarchy
(PoA) and Price of Stability (PoS) are defined as
$$\mbox{PoA}=\max_{s\in \mathcal{N}(\mathcal{G})}
\frac{\texttt{cost}(s)}{OPT}\qquad\qquad\qquad \mbox{PoS}=\min_{s\in
\mathcal{N}(\mathcal{G})} \frac{\texttt{cost}(s)}{OPT}.$$

\subsection{Covering Game}\label{sec:defsng}

Given agents $[n]=\{1,2,\dots,n\}$, a collection of sets
$\F\subseteq 2^{[n]}$, costs $c_i$ for $i\in[n]$, and weights
$w_\sigma$ for $ \sigma\in\F$, we describe the covering game
$\mathcal G=\ang{[n],(S_i),(\texttt{cost}_i)}$ where actions
$S_i=\{\on, \off\,\}$ and cost $\texttt{cost}_i$ as defined in
\eqref{costeqn} for every agent $i\in[n]$. We let $\Delta_k$ be the
`$k$-th order' maximum degree of the hypergraph induced by sets.
Namely,
$$\Delta_k=\Delta_k(\mathcal G):=
\max_{i_1,\dots,i_k\in [n]}\left|\{ \sigma \in \F :
\{i_1,\dots,i_k\}\subset\sigma,i_i\neq i_j,\forall i\neq
j\}\right|.$$ In addition, we define
\begin{equation*}
c_{\max}:=\max_{i\in[n]} c_i\qquad c_{\min}:=\min_{i\in[n]}
c_i\qquad w_{\max}:=\max_{\sigma\in\F} w_\sigma\qquad
w_{\min}:=\min_{\sigma\in\F} w_\sigma.
\end{equation*}

Before defining the cost functions, we introduce some notation. We
say a set $\sigma\in \F$ is `covered' in joint strategy $s$ if
$s_i=\on$ for some $i\in\sigma$. Otherwise, $\sigma$ is said to be
`uncovered'. Denote the collection of sets that include agent $i$
and are uncovered in $s$ with $\mathcal F_i^{u}(s)$, or simply
$\F_i^u$ when $s$ is clear from context.  The entire set of
uncovered sets is written $\F^u=\bigcup_{i\in[n]}\F_i^u$.  For
$\sigma\subseteq[n]$, define $c(\sigma):=\sum_{i\in\sigma} c_i$, and
for $\F'\subseteq\F$, define $w(\F'):=\sum_{\sigma\in\F'}w_\sigma$.
Now define the cost function of agent $i$ is defined with respect to
any joint strategy $s\in S$ as follows:
\begin{equation}\texttt{cost}_i(s)=
\begin{cases}
\quad c_i &\mbox{if}~ s_i=\mbox{\on}\vspace{0.05in}\\
w(\F_i^u) &\mbox{if}~ s_i=\mbox{\off}\vspace{0.05in}.
\end{cases}\label{costeqn}\end{equation}
Observe that $c_i$ expresses how much agent $i$ prefers to cover the
sets containing $i$. For example, if $c_{\max}<w_{min}$, then each
agent prefers to avoid the situation that there exists an uncovered
sets containing her. As we explain in Section \ref{packing}, these
covering games can be interpreted as equivalent packing games.

For a joint action (or strategy profile) $s\in S$, let
$\mbox{ON}(s)$ and $\mbox{OFF}(s)$ be sets of nodes that are \on\
and \off, respectively. It is easy to check that the social cost
has the following simple form:
\begin{equation}\label{eq:cost}
\texttt{cost}(s)=\sum_{i\in [n]}\texttt{cost}_i(s)
=c(\mbox{ON}(s))+\sum_{\sigma\in \mathcal F^u(s)} |\sigma|\cdot
w(\sigma).
\end{equation}

\paragraph{Best response convergence.} Recall that best response dynamics converge to pure Nash equilibria for potential games.  Now observe that the covering game is an exact potential game with potential function
\begin{equation}\Phi(s)=c\left(\mbox{ON}(s)\right)+w(\F^u(s)).\label{eq:phi}\end{equation} Combining this observation with the social cost formula implies that for any $s\in S$ we have
\begin{equation}\Phi(s)\leq\texttt{cost}(s)\leq F_{\max}\cdot\Phi(s),\label{eq:phiandcost}\end{equation}
where we let $F_{\max}$ be the size of the largest set i.e.\
$F_{\max}=\max_{\sigma\in \F} |\sigma|$. 

\paragraph{Optimization and equilibrium quality.}
The star graph example from the introduction reveals that PoA in the
covering game can be very large. More generally, certain covering
game instances exist with PoA $\Omega\left({n}\right)$ even
restricted to the simple case $\Delta_2=1$.\footnote{For example,
let $c_i=c$ for all $i\in[n]$ and let $w_\sigma=1$ for all
$\sigma\in\F$.  Label an arbitrary set of $\ceil{c}$ elements $L$,
and label the other elements $R$. Define $\F$ to be all sets with
one element in $L$ and one in $R$.  It is straightforward to check
that the solution with all $L$ \on\ and all $R$ \off\ is a Nash
equilibrium with cost $c\cdot \ceil{c}$, while the solution with all
$L$ \off\ and all $R$ \on\ is a Nash equilibrium with cost
$c\cdot(n-\ceil{c})$.} This motivates the need for efficient
dynamics with better guarantees than convergence to arbitrary
equilibria.

As a step in that direction, here we provide a centralized
LP-rounding-based poly-time algorithm to find a low-cost
configuration $s^{ad}$ for the covering game as follows.
\begin{itemize}
\item[1.] Solve the following Linear Programming (LP), and obtain
the solution $x^*$.
\begin{equation}
\texttt{minimize}~~\sum_{i=1}^n c_i\cdot x_i\qquad\texttt{subject
to}\qquad \sum_{i\in \sigma} x_i\geq 1~~\forall \sigma\in \mathcal
F,~x_i\in [0,1]\label{ILP}
\end{equation}
\item[2.] Set $\begin{cases}s^{ad}_i=\mbox{\on}&\mbox{if}~  x^*_i\geq 1/F_{\max}\\
s^{ad}_i=\mbox{\off}&\mbox{otherwise}\end{cases}.$
\end{itemize}
The following lemma proves that the algorithm is a $F_{\max}
\ceil{c_{\max}/w_{\min}}$-approximation one for minimizing the
social cost \eqref{eq:cost}.
\begin{lemma}
The configuration $s^{ad}$ obtained from the algorithm has
$$\texttt{cost}(s^{ad})~\leq~F_{\max}
\ceil{c_{\max}/w_{\min}}\cdot OPT,$$ where we recall that
$OPT=\min_s \texttt{cost}(s)$.
\end{lemma}
\begin{proof}
Let $s^*$ be the optimal configuration i.e.\
$\texttt{cost}(s^*)=OPT$. If there exists a uncovered set $\sigma$
under the configuration $s^*$, choose one element from $\sigma$ and
force it to be turned \on. Repeat this procedure until all sets are
covered, and say $s^{\dag}$ be the resulting configuration. Now
observe that $\texttt{cost}(s^\dag)\leq
\ceil{c_{\max}/w_{\min}}\cdot OPT$. Therefore, it follows that
$$\texttt{cost}(s^{ad})~\leq~F_{\max}\cdot\sum_i c_i\cdot x_i~\leq~F_{\max}\cdot
\texttt{cost}(s^\dag)~\leq~ F_{\max} \ceil{c_{\max}/w_{\min}}\cdot
OPT.$$
\end{proof}
Under the assumption \eqref{eq:assume}, this is an
$O(F_{\max})\cdot OPT$-approximation algorithm to the optimal social
cost.

\section{Public Service Advertising} 
\label{psa} In this section and the following one, we show that
price of anarchy is avoidable in covering games even using best
response-inspired dynamics as long as these dynamics incorporate
some form of suggestion from a weak central authority that is aware
of a high quality equilibrium.  

The first model we study in this paper is the public service
advertising (PSA) model in \cite{BBM} in which a central authority
broadcasts a strategy for each agent, which some agents receive and
temporarily follow.  Player behavior is described in two phases:
\begin{itemize}
\item[{\bf 1:}] Play begins in an arbitrary state, and a central authority advertises joint action $s^{ad}\in S$.  Each agent receives the proposed strategy independently with probability $\alpha\in(0,1)$.  Agents that receive this signal are called receptive.  Receptive agents play their advertising strategies throughout Phase 1, and non-receptive agents undergo best response dynamics to settle on a joint strategy that is a Nash equilibrium given the fixed behavior of receptive agents.  We call this joint strategy $s'$.
\item[{\bf 2:}] All agents participate in best response dynamics until convergence to some Nash equilibrium $s''$.
\end{itemize}
Since our covering game is a potential game and all potential games
eventually converge to a Nash equilibrium under best response
dynamics, both phases are guaranteed to terminate.  Furthermore,
convergence occurs in poly-time with respect to parameters $\{n,
c_1,\dots,c_n, w(\sigma):\sigma\in \mathcal F\}$.\footnote{This is
because $\Phi$ is bounded above and below by functions of these
parameters and decreases under best response dynamics.}
\subsection{Effect of Advertising in PSA}

In this section we show that advertising helps significantly in
covering games. In particular, we show that if the advertising
strategy $s^{ad}$ has low social cost, then the cost of the
resulting equilibrium is low even if only a small constant $\alpha$
fraction of the agents receive and respond to the signal. Theorem
\ref{PSAoutcome-1} formalizes the general result of this section,
and Theorem \ref{PSAoutcome-2} improves this result for particular advertising strategies. 
For the convenience, in this section we assume costs and weights are
bounded above and below, i.e.\
\begin{equation}\label{eq:assume}
c_{\max}:=\max_{i\in[n]} c_i=O(1)\quad c_{\min}:=\min_{i\in[n]}
c_i=\omega(1)\quad w_{\max}:=\max_{\sigma\in\F} w_\sigma=O(1)\quad
w_{\min}:=\min_{\sigma\in\F} w_\sigma=\omega(1).
\end{equation}

\begin{theorem}\label{PSAoutcome-1}
For any advertising strategy $s^{ad}$ in the PSA model,
\begin{equation}
\E[\texttt{cost}(s'')]\leq
\begin{cases}
\displaystyle O(1)\cdot\texttt{cost}(s^{ad}) &\mbox{if}~F_{\max}=2\\
\displaystyle O(\Delta_2)\cdot\texttt{cost}(s^{ad})^2
&\mbox{if}~F_{\max}=O(1)
\end{cases}. \end{equation}
\end{theorem}
Theorem \ref{PSAoutcome-1} implies that if $s^{ad}$ is obtained from
the $O(F_{\max})$-approximation poly-time algorithm described in
Section \ref{sec:defsng}, the following corollary holds.
\begin{corollary}\label{cor1}
There exists a poly-time algorithm to find an advertising strategy
$s^{ad}$ for the PSA model such that
\begin{equation} \E[\texttt{cost}(s'')]\leq \begin{cases}
O(1)\cdot OPT &\mbox{if}~F_{\max}=2\\O(\Delta_2)\cdot
OPT^2&\mbox{if}~F_{\max}=O(1)\end{cases}.\notag\end{equation}
\end{corollary}
\paragraph{Effective advertising.} We additionally consider advertising strategies particular to our game for improved performance of the model.  We say that advertising strategy $s^{ad}$ satisfies condition ($\star$) if
\begin{equation}
\left\lfloor\frac{c_{\max}}{w_{\min}}\right\rfloor x^{\left\lfloor
{c_{\max}} / {w_{\min}}\right\rfloor} \left(1-\alpha^{\mathcal
F_{\max}}\right)^{x-\left\lfloor {c_{\max}} /
{w_{\min}}\right\rfloor}~\leq~ \frac1{n^2},\qquad\text{ for all
}~~x\geq\frac{\Delta^*_1}{\Delta_2(F_{\max}-1)},\tag{$\star$}\end{equation}
where $\Delta^*_1:=\Delta^*_1(s^{ad})$ is the smallest number of
sets containing a given \on\ element in $s^{ad}$ as the unique \on\
element. We say $\Delta^*_1$ is the `core' minimum degree of \on\
elements in $s^{ad}$. Intuitively, the condition $(\star)$ means
that each \on\ element in the advertising strategy $s^{ad}$ `solely'
contributes a large number of sets to cover. We establish the
following stronger theorem which implies that agents will reach a
state of social cost $O(\texttt{cost}(s^{ad}))$ at the end of Phase
2 if $s^{ad}$ satisfied the condition $(\star)$.


\begin{theorem}\label{PSAoutcome-2}
For an advertising strategy $s^{ad}$ satisfying the condition
$(\star)$ in the PSA model,
$$\texttt{cost}(s'')= O(1)\cdot\texttt{cost}(s^{ad})\qquad\mbox{with probability}~~1-\frac1n,
\qquad\mbox{if}~F_{\max}=O(1).$$
\end{theorem}
The following corollary implies that it is possible to find such  an
advertising strategy $s^{ad}$ of low cost.
\begin{corollary}\label{cor2}
There exists a poly-time algorithm to find an advertising strategy
$s^{ad}$ for the PSA model such that
$$\texttt{cost}(s'')= O(\Delta_2\log n)\cdot OPT\qquad\mbox{with probability}~~1-\frac1n,
\qquad\mbox{if}~F_{\max}=O(1).$$
\end{corollary}
\begin{proof}
Here we explain how to find an advertising strategy $s^{ad}$
satisfying the condition $(\star)$ as well as being of low cost.
Observe that any joint strategy $s$ with
$\Delta^*_1=\Delta^*_1(s)\geq B \Delta_2 \log n$ for a large enough
constant $B$ (depending on constants $c_{\max}/w_{\min}$, $\alpha$,
$F_{\max}$) satisfies the condition $(\star)$. Then starting from
the joint strategy with social cost $O(1)\cdot OPT$ obtained from
the algorithm in Section \ref{sec:defsng}, one can greedily
construct a joint strategy $s^{ad}$ satisfying the condition
$(\star)$ with social cost $O(\Delta_2 \log n)\cdot OPT$ (greedily
turning \off\ every agent that is the unique \on\ element in fewer
than $B \Delta_2 \log n$ sets). For the advertising strategy
$s^{ad}$ satisfying the condition $(\star)$ as constructed above,
the conclusion of Corollary \ref{cor2} follows from Theorem
\ref{PSAoutcome-2}.
\end{proof}

\subsection*{Proof of Theorem
\ref{PSAoutcome-1}}\label{sec:pfofthmPSAoutcome-1} From
\eqref{eq:phiandcost} and $F_{\max}=O(1)$, any sequence of best
response moves increases social cost by at most a constant factor.
All agents best respond in Phase 2, and hence
$\texttt{cost}(s'')=O(\texttt{cost}(s'))$. It suffices to bound
$\texttt{cost}(s')$. At a high level, we do this by providing a
bound (i.e.\ Lemma \ref{PSAEf-1}) on the total weight of uncovered
sets that are not uncovered in $s^{ad}$ and then we give a bound
(i.e.\ Lemma \ref{PSARf-1}) on the number of agents that are \off\
in $s^{ad}$ but \on\ at the end of Phase 1.

First, let us introduce some notation. We say two agents contained
in a common set are neighbors. Let $L$ and $R$ denote the set of
agents that are \on\ and \off\ in
$s^{ad}$, respectively. 
Let $L_\off$ (and $R_\on$) denotes the set of agents in $L$ (and
$R$), who are \off\ (and \on) in $s'$. Let $\F_R$ denote the
collection of sets uncovered in $s^{ad}$, and let $\F_{bad}$ denote
the collection of sets not in $\F_R$ but uncovered in $s'$. Then
from \eqref{eq:cost}, \eqref{eq:assume} and $F_{\max}=O(1)$, we have
\begin{align*}
\E[\texttt{cost}(s')]&~\leq~ \texttt{cost}(s^{ad})+
\E[c(R_\on)]+F_{\max}\cdot E[w(\F_{bad})]\\
&~=~ \texttt{cost}(s^{ad})+ O\left(\E[|R_\on|]\right)+O(
E[w(\F_{bad})]),
\end{align*} where we note that
$$\texttt{cost}(s^{ad})~\geq~ c(L) + w(\F_R)~=~\Omega(|L|) + \Omega( |\F_R|).$$
 Therefore, the following two lemmas bounding $w(\F_{bad})$ and
$|R_\on|$ leads to the desired bound on $\texttt{cost}(s')$, which
completes the proof of Theorem \ref{PSAoutcome-1} from
$\texttt{cost}(s'')=O(\texttt{cost}(s'))$.

\begin{lemma}\label{PSAEf-1}
$w(\F_{bad})\leq c(L)$. \end{lemma}
\begin{proof}
 Each set in $\F_{bad}$ should contain an \off\ element in
$L_\off$ that is best responding in $s'$. Hence,
$$w(\F_{bad})~=~w\Big(\bigcup_{\ell\in L_\off} \F^u_{\ell}\Big)~\leq~
\sum_{\ell\in L_\off}w( \F^u_{\ell})~\leq~ \sum_{\ell\in L_\off}
c_\ell~\leq~ c(L),$$ where the second inequality is from the fact
that $\ell$ is best responding (i.e.\ its cost exceeds the total
weight of uncovered sets including it since it chooses to be \off).
This completes the proof of Lemma \ref{PSAEf-1}.
\end{proof}

\begin{lemma}\label{PSARf-1}
$\E[|R_\on |]\leq
\begin{cases}
|\F_R|+\Delta_2\cdot |L|^2+O(\Delta_2)\cdot |L| &\mbox{if}~F_{\max}=O(1)\\
|\F_R|+O(|L|) &\mbox{if}~F_{\max}=2
\end{cases}$.
\end{lemma}
\begin{proof}

Since each element $r$ in $R_\on$ plays best response in $s'$, $r$
should be contained in a set $\sigma_r$ as the unique \on\ element.
We define disjoint sets $R_\on^{(1)}$ and $R_\on^{(2)}$ such that
$$R_\on~=~R_\on^{(1)}\cup R_\on^{(2)}$$ and $r\in R_\on^{(1)}$ if
$\sigma_r \in\F_R$. By definition of $R_\on^{(1)}$, it easily
follows that
\begin{equation}
|R_\on^{(1)}|~\leq~|\F_R|.\label{eq:PSARf-4}
\end{equation}

Now consider $R_\on^{(2)}$. Let $\F_\off$ be the collection of
`left' uncovered sets, i.e.\ $\sigma\in\F_\off$ if $\sigma \cup L$
is a non-empty subset of $L_\off$. Hence, by definition of
$R_\on^{(2)}$, $\sigma_r$ is in $\F_\off$ for each $r\in
R_\on^{(2)}$. This implies that
\begin{equation}
|R_\on^{(2)}|~\leq~|\F_\off|.\label{eq:PSARf-5}
\end{equation}
We let $\F_\off^*\subset \F_\off$ be the collection of sets
containing a unique element in $L_\off$. Then, we have
\begin{equation}
|\F_\off\setminus \F_\off^*|~\leq~ \begin{cases} \Delta_2\cdot
|L|^2&\mbox{if}~~F_{\max}=O(1)\\
\qquad 0&\mbox{if}~~F_{\max}=2
\end{cases}.\label{eq:PSARf-6}
\end{equation}
This is because the number of sets with more than one element in
$L_\off$ is bounded by $\Delta_2\cdot |L|^2$ (remember that each
pair of agents is contained in at most $\Delta_2$ common sets).
Clearly, there are no such sets when $F_{\max}=2$.

We now bound the expected size of $\F_\off^*$. It follows that
\begin{equation}
\E[|\F_\off^*|]~=~ \sum_{\ell\in L}|\F_\ell^*|\cdot\Pr[\ell\text{ is
\off}],\label{eq:PSARf-2}\end{equation} where we let $\F_\ell^*$ be
the collection of sets including $\ell$ as the unique element in
$L$. Further, we observe that
\begin{align}
\Pr[\ell\text{ is \off}]
&~\leq~ \Pr[|\{\rho\in\F_\ell^*: \text{ all }\rho\cap R\text{ are \off}\}|\leq c_{\max}/w_{\min}]\notag\\
&~\leq~ \Pr[|\{\rho\in\F_\ell^*: \text{ all }\rho\cap R\text{ are receptive}\}|\leq c_{\max}/w_{\min}]\notag\\
&~\leq~ \Pr[|\{\rho\in\widehat\F_\ell^* : \text{ all }\rho\cap
R\text{ are receptive}\}|\leq c_{\max}/w_{\min}], \notag
\end{align}
where we define $\widehat\F_\ell^*\subseteq \F_\ell^*$ such that no
pair of sets in $\widehat \F_\ell^*$ have common elements in $R$ and
the size of $\widehat\F_\ell^*$ is not too small i.e.\
$|\F_\ell^*|\leq (F_{\max}-1)\Delta_2\cdot|\widehat \F_\ell^*|$.
From definitions of $F_{\max}$ and $\Delta_2$, the existence of such
set $\widehat\F_\ell^*$ follows. Since no pair of sets in $\widehat
\F_\ell^*$ have common elements in $R$, the events that all
$\rho\cap R$ are receptive for $\rho \in \widehat \F_\ell^*$ become
independent with each other and each happens with probability at
least $\alpha^{F_{\max}}$. Therefore,
\begin{align}
\Pr[\ell\text{ is \off}] &~\leq~ \Pr[|\{\rho\in\widehat\F_\ell^* :
\text{ all }\rho\cap R\text{ are receptive}\}|\leq
c_{\max}/w_{\min}]\notag\\
 &~\leq~ \sum_{i=0}^{\lfloor
c_{\max}/w_{\min}\rfloor} {|{\widehat \F}_\ell^{*}| \choose
i}\left(1-\alpha^{F_{\max}}\right)^{|{\widehat
\F}_\ell^{*}|-i}\left(\alpha^{F_{\max}}\right)^i. \label{eq:PSARf-1}
\end{align}
Combining \eqref{eq:PSARf-2} and \eqref{eq:PSARf-1} implies that
\begin{align}
\E[|\F_\off^*|]& ~\leq~ \sum_{\ell\in L}|\F_\ell^*|\cdot
\sum_{i=0}^{\lfloor c_{\max}/w_{\min}\rfloor} {|{\widehat
\F}_\ell^{*}| \choose i}\left(1-\alpha^{F_{\max}}\right)^{|{\widehat
\F}_\ell^{*}|-i}\left(\alpha^{F_{\max}}\right)^i\notag\\
&~\leq~ {(F_{\max}-1)\Delta_2}\sum_{\ell\in L} \sum_{i=0}^{\lfloor
c_{\max}/w_{\min}\rfloor} |\widehat \F_\ell^*|{|{\widehat
\F}_\ell^{*}| \choose i}\left(1-\alpha^{F_{\max}}\right)^{|{\widehat
\F}_\ell^{*}|-i}\left(\alpha^{F_{\max}}\right)^i\notag\\
&~=~ O(\Delta_2)\cdot |L|,
\label{eq:PSARf-3}
\end{align}
where the last equality is from the following proposition of which
proof is presented in Appendix \ref{sec:pftech}.
\begin{proposition}\label{tech}
For constant $a\in(0,1)$ and $0< c\leq d$,
\[
\sum_{i=0}^{\floor c} d\binom{d}{i}(1-a)^{d-i} \alpha^i=O(\lceil
c\rceil).
\] \end{proposition}
Finally, combining \eqref{eq:PSARf-4}, \eqref{eq:PSARf-5},
\eqref{eq:PSARf-6} and \eqref{eq:PSARf-3} leads to the desired
conclusion of Lemma \ref{PSARf-1}, where note that $\Delta_2=1$ when
$F_{\max}=2$.
\end{proof}

\subsection*{Proof of Theorem \ref{PSAoutcome-2}}

We will use the same notation $R_\on$ and $\F_{bad}$ as in the proof
of Theorem \ref{PSAoutcome-1}. As we explain in the proof of Theorem
\ref{PSAoutcome-1}, it suffices to prove that the social cost at the
end of Phase 1 is $O(\texttt{cost}(s^{ad}))$ with probability
$1-1/n$.

To this end, the following lemma establishes the condition $(\star)$
ensures that all agents in $L$ turn \on\ with probability $1-1/n$ at
the end of Phase 1. Under the event, only sets in $\F_R$ are
uncovered and the additional social cost incurred by agents in
$R_\on$. The lemma shows that such additional cost is at most
$\texttt{cost}(s^{ad})$. Hence, this completes the proof of Theorem
\ref{PSAoutcome-2}.

\begin{lemma}\label{PSARf-2}
If the advertising strategy $s^{ad}$ satisfies the condition
$(\star)$, then
$$\F_{bad}=\emptyset\quad\mbox{and}\quad c(R_\on)\leq w(\F_R)\qquad\mbox{with probability}~~1-\frac1n.$$
\end{lemma}
\begin{proof}
We will use the same notation in the proof Lemma \ref{PSARf-1}. As
in the proof Lemma \ref{PSARf-1}, for any $\ell\in L$ there is some
subset $\widehat \F_\ell^*\subseteq \F_\ell^*$ such that no pair of
sets in $\widehat\F_\ell^*$ have common elements in $R$ and
$|\widehat\F_\ell^*|\geq
\frac{|\F_\ell^*|}{(F_{\max}-1)\Delta_2}\geq
\frac{\Delta^*_1}{(F_{\max}-1)\Delta_2}$.
Then as we derived in in the proof Lemma \ref{PSARf-1},
\begin{align*}
\Pr[\ell\text{ is \off}]&~\leq~  \sum_{i=0}^{\lfloor
c_{\max}/w_{\min}\rfloor}{|\widehat\F_\ell^*| \choose
i}\left(1-\alpha^{F_{\max}}\right)^{|\widehat\F_\ell^*|-i}\left(\alpha^{F_{\max}}\right)^i~\leq
\frac1{n^2},
\end{align*}
where the last inequality is from the condition $(\star)$. From the
union bound, $\Pr[L_\off=\emptyset]\geq 1-1/n$ and hence
$\F_{bad}=\emptyset$.

Now assume the event that all nodes in $L$ are \on. Observe that for
each best responding $r\in R_\on$, $c_r$ is no greater than the
total weight of all sets containing $r$ as the unique \on\ agent.
Since we assume all nodes in $L$ are \on\, these sets are a subset
of $\F_R$. Further, since there is no overlap in these sets between
different agents in $R_\on$, we can sum over all $r\in R_\on$ to
derive $c(R_\on)\leq w(\F_R)$. This completes the proof of Lemma
\ref{PSARf-2}.

\end{proof}

\subsection{Extension to Unbounded Costs and Weights}


All the results and proof techniques in this paper naturally extend
to general weights and costs. In particular, one can obtain the
following theorem (analogous to Theorem \ref{PSAoutcome-1}) in the
PSA model without the assumption \eqref{eq:assume} via calculating
explicit quantities in each step in the proof of Theorem
\ref{PSAoutcome-1}.

\begin{theorem}
For any advertised strategy $s^{ad}$ in the PSA model,
\begin{equation}
\E[\texttt{cost}(s'')]\leq
\begin{cases}
\displaystyle O(\Delta_2\cdot\ceil{c_{\max}/w_{\min}}\cdot
c_{\max}/c_{\min}^3)\cdot\texttt{cost}(s^{ad})^2 &\mbox{if}~F_{\max}=O(1)\\
\displaystyle O(\ceil{c_{\max}/w_{\min}}\cdot
c_{\max}/c_{\min})\cdot\texttt{cost}(s^{ad}) &\mbox{if}~F_{\max}=2
\end{cases}. \end{equation}
\end{theorem}

\section{Learn-then-decide}\label{ltd}
We now study the set cover game in the learn-then-decide (LTD) model
of ~\cite{bbm-ics-10}. In contrast to PSA, agents in LTD are neither
strictly receptive nor strictly best responders in the initial
exploration phase, but they choose one of these options for the
final exploitation phase:
\begin{itemize}
\item[{\bf 1:}] Play begins in an arbitrary state, and a central authority advertises joint action $s^{ad}\in S$.  Player $i$ is associated with fixed probability $p_i\geq \beta\in(0,1)$. Agents are chosen to update uniformly at random for each of $T^*$ time steps. When $i$ updates, he plays $s^{ad}_i$ with probability $p_i$ or best response with probability $1-p_i$. The state at time $T^*$ is denoted $s'$.
\item[{\bf 2:}] At time $T^*$, all agents in random order individually commit arbitrarily to $s^{ad}_i$ or best response. Then agents take turns in random order playing their chosen strategy until they reach a Nash equilibrium $s''$ given the fixed behavior of $s^{ad}$ followers.
\end{itemize}

\subsection*{Effect of Advertising in LTD}
For the convenience, in this section we again assume costs and
weights are bounded above and below, i.e.\ the assumption
\eqref{eq:assume}.
The following result in the LTD model is analogous to Theorem
\ref{PSAoutcome-1} in the PSA model.
\begin{theorem}\label{LTDoutcome-1}
There exists a $T^*\in poly(n)$ such that for any advertising
strategy $s^{ad}$ in the LTD model,
\begin{equation}\label{eq1:mainltd}
E[\texttt{cost}(s'')]\leq
\begin{cases}
\displaystyle O(\Delta_2^2)\cdot\texttt{cost}(s^{ad})^2 &\mbox{if}~F_{\max}=O(1)\\
\displaystyle O(1)\cdot\texttt{cost}(s^{ad}) &\mbox{if}~F_{\max}=2
\end{cases}.
\end{equation} 
\end{theorem}
Theorem \ref{LTDoutcome-1} implies that if $s^{ad}$ is obtained from
the $O(F_{\max})$-approximation poly-time algorithm described in
Section \ref{sec:defsng}, the following corollary holds.
\begin{corollary}\label{cor1}
There exists a poly-time algorithm to find an advertising strategy
$s^{ad}$ for the LTD model such that
\begin{equation} \E[\texttt{cost}(s'')]\leq \begin{cases}
O(\Delta_2^2)\cdot OPT^2&\mbox{if}~F_{\max}=O(1)\\O(1)\cdot OPT
&\mbox{if}~F_{\max}=2\end{cases}.\notag\end{equation}
\end{corollary}

\subsection*{Proof of Theorem \ref{LTDoutcome-1}}\label{sec:LTDoutcome-1}
To begin with, we note that while LTD differs from PSA in both
phases, the proof that cost is low in Phase 1 of LTD is very similar
to the proof of Theorem \ref{PSAoutcome-1}. However, showing that
the cost stays low in Phase 2 imposes new challenges.

We will use the same notation as in the proof of Theorem
\ref{PSAoutcome-1}. We first define $\mathcal E=\mathcal E(T',T^*)$
for $1<T'< T^*$ as the event that every element in $L$ updates at
least once before time $T'$ {\em after} every element in $R$ has
updated at least once, and then every element in $R$ again updates
at least once at some time $t\in[T',T^*]$. Clearly there exist some
$T',T^*\in poly(n)$ such that $\mathcal E=\mathcal E(T',T^*)$
happens with the following high probability, i.e.\
$$\Pr[\mathcal E]~\geq~1-\frac1{n^{F_{\max}}}.$$
Then, we have
\begin{eqnarray}
E[\texttt{cost}(s'')]&=& \Pr[\mathcal E]\cdot
E[\texttt{cost}(s'')~|~\mathcal E]+ \Pr[\mathcal E^c]\cdot
E[\texttt{cost}(s'')~|~\mathcal E^c]\notag\\
&\leq&E[\texttt{cost}(s'')~|~\mathcal E]+ \frac1{n^{F_{\max}}}
\cdot O\left(n^{F_{\max}}\right)\notag\\
&=&E[\texttt{cost}(s'')~|~\mathcal E]+ O(1),\label{eq9}
\end{eqnarray}
where for the inequality we use the fact that the social cost is
always bounded above by $c_{\max} \cdot n + F_{\max}\cdot |\F|=
O\left(n^{F_{\max}}\right)$ from \eqref{eq:cost} and
\eqref{eq:assume}.

Therefore, it suffice to bound $E[\texttt{cost}(s'')~|~\mathcal E]$,
where our choice of $T^*$ is primarily for guaranteeing that
$\mathcal E$ happens with such a high probability. We first bound
the expected social cost at the end of Phase 1 under the event
$\mathcal E$ as below. And later, we will bound the increase in the
social cost in Phase 2.
\begin{lemma}\label{LDPh1}

\begin{equation*}
E[\texttt{cost}(s')~|~\mathcal E]\leq
\begin{cases}
O(1)\cdot\texttt{cost}(s^{ad})
&\mbox{if}~F_{\max}=2\\
O(\Delta_2)\cdot\texttt{cost}(s^{ad})^2 &\mbox{if}~F_{\max}=O(1)
\end{cases}.
\end{equation*}
\end{lemma}

\begin{proof}

Similarly as in the proof of Theorem \ref{PSAoutcome-1}, we again
note that
\begin{eqnarray}
\texttt{cost}(s')&=& \texttt{cost}(s^{ad})+O(|R_\on|) +
O(w(\F_{bad}))\label{eq0-0}\\\texttt{cost}(s^{ad})&=& \Omega( |L|) +
\Omega(|\F_R|).\label{eq0-1}
\end{eqnarray} We again remind that we will use the same notation
as in the proof of Theorem \ref{PSAoutcome-1}.

Hence, it suffices to bound $w(\F_{bad})$ and $|R_\on|$ in terms of
$|L|$ and $|\F_R|$. First consider $w(\F_{bad})$. We separately
analyze the weights of two types of $\F_{bad}$. First consider a set
in $\F_{bad}\cap 2^L$, i.e.\ a set consisting only of elements in
$L_\off$. Suppose we attribute the weight of such a set to its
element $\ell$ that updated most recently before the end of Phase 1.
Because $\ell\in L_\off$ played best response most recently, the
weight of all sets in $\F_{bad}\cap 2^L$ attributed to $\ell$ is at
most $c_\ell$. Summing over all $\ell\in L_\off\subseteq L$ gives
\begin{equation}
w(\F_{bad}\cap 2^L)\leq c(L)=O(|L|).\label{eq1}
\end{equation}

Now consider a set in $\F_{bad}\backslash 2^L$, i.e.\ a set which
has elements in both $L_\off$ and $R$ and all of them are \off\ at
the end of Phase 1. By definition of $\mathcal F_\off$,
$\F_{bad}\backslash 2^L\subset \mathcal F_\off$. Under assuming the
event $\mathcal E$, the proof arguments to bound $|\mathcal F_\off|$
in the proof of Lemma \ref{PSARf-1} identically work in the LTD
model (using $\beta$ instead of $\alpha$), i.e.\ we have
\begin{equation}\label{eq2}
E\left[\left|\F_{bad}\backslash 2^L\right|~|~\mathcal E\right]\leq
E\left[|\mathcal F_\off|~|~\mathcal E\right]=
\begin{cases}
\Delta_2\cdot |L|^2+O(\Delta_2)\cdot |L| &\mbox{if}~F_{\max}=O(1)\\
O(|L|) &\mbox{if}~F_{\max}=2
\end{cases}.
\end{equation}
From \eqref{eq1} and \eqref{eq2}, it follows that
\begin{equation}\label{eq3}
E\left[w(\F_{bad})~|~\mathcal E\right]=
\begin{cases}
O(\Delta_2\cdot |L|^2)+O(\Delta_2)\cdot |L| &\mbox{if}~F_{\max}=O(1)\\
O(|L|) &\mbox{if}~F_{\max}=2
\end{cases}.
\end{equation}

Now under assuming event $\mathcal E$, one can observe that the
conclusion and proof strategy of Lemma \ref{PSARf-1} also works for
$|R_\on|$ in the LTD model, i.e.\
\begin{equation}\label{eq4}
E\left[|R_\on|~|~\mathcal E\right]=
\begin{cases}
|\F_R|+O(\Delta_2\cdot |L|^2)+O(\Delta_2)\cdot |L| &\mbox{if}~F_{\max}=O(1)\\
|\F_R|+O(|L|) &\mbox{if}~F_{\max}=2
\end{cases}.
\end{equation}

Therefore, combining these bounds \eqref{eq0-0}, \eqref{eq0-1},
\eqref{eq3} and \eqref{eq4} leads to the desired bound of Lemma
\ref{LDPh1}.
\end{proof}

We now bound the cost increase in Phase 2 assuming $\mathcal E$.
From \eqref{eq:phiandcost} and $F_{\max}=O(1)$ it suffices to
provide a bound on the expected increase in the potential function
throughout Phase 2, i.e.\
\begin{eqnarray}
\texttt{cost}(s'')&\leq&F_{\max}\cdot
\Phi(s'')+F_{\max}\cdot(-\Phi(s') +\texttt{cost}(s'))\notag\\
&=& O(\Phi(s'')-\Phi(s'))+O(\texttt{cost}(s')).\label{eq10}
\end{eqnarray}
The following lemma bounds the expected potential increase
$\Phi(s'')-\Phi(s')$ under assuming event $\mathcal E$. Finally,
combining \eqref{eq9}, \eqref{eq10}, Lemma \ref{LDPh1} and Lemma
\ref{LDstay} lead to the desired conclusion of Theorem
\ref{LTDoutcome-1}.

\begin{lemma}\label{LDstay} 
\begin{align*}
\E[\Phi(s'')-\Phi(s')\mid\mathcal E]\leq \begin{cases}
O(\Delta_2^2)\cdot\texttt{cost}(s^{ad})^2 &\mbox{if}~F_{\max}=O(1)\\
O(1)\cdot\texttt{cost}(s^{ad}) &\mbox{if}~F_{\max}=2
\end{cases}.\end{align*}
\end{lemma} 

\begin{proof}
Since best response moves do not increase the potential function
$\Phi$, we only consider updates of agents following the advertising
strategy $s^{ad}$ in Phase 2.
Since each `$s^{ad}$ follower' changes strategies at most once in
Phase 2, it suffices to consider a single \off-\on\ move (following
$s^{ad}$) for each agent in $L$ and a single \on-\off\ move
(following $s^{ad}$) for each agent in $R_\on$. For each $\ell\in
L$, an \off-\on\ movie (i.e.\ $\ell$ changes his decision from \off\
to \on) increases potential by at most $c_\ell$. Hence,
\begin{equation}\label{eq5}
 \mbox{the total potential
increase by \on-\off\ moves is at most $c(L)=O(|L|)$.}\end{equation}

Now consider another type of moves, i.e.\ a single \on-\off\ move
for each agent in $R_\on$. For each $r\in R_\on$ that first turns
\off\ at time $t\geq T^*$, let $\F_r$ be the collection of sets
containing $r$ such that all of their other elements are \off\ at
time $t$. Then the potential increases by at most
$w(\F_r)=O(|\F_r|)$ at time $t$. Hence,
\begin{equation}\label{eq6} \mbox{the total potential
increase by \off-\on\ moves is at most $O(\sum_{r\in
R_\on}|\F_r|)=O\left(\left|\cup_{r\in R_\on}
\F_r\right|\right)$.}\end{equation}

We will bound the expectation of $|\cup_{r\in R_\on} \F_r|$. To this
end, we consider two types of set $\sigma$ including an element in
$R_\on$: (a) $\sigma$ has an element $\ell_\sigma\in L$ that was
\on\ at the end of Phase 1, and (b) otherwise. Observe that the
expected number of sets of type (b) which does not in $\F_R$ is
already bounded in the proof of Lemma \ref{LDPh1} by $|\F_\off|$
i.e.\
\begin{equation*}
E[\,\mbox{the number of sets in $\cup_{r\in R_\on} \F_r\setminus
\F_R$ of type (b)}~|~\mathcal E]=
\begin{cases}
\Delta_2\cdot |L|^2+O(\Delta_2)\cdot |L| &\mbox{if}~F_{\max}=O(1)\\
O(|L|) &\mbox{if}~F_{\max}=2
\end{cases}.
\end{equation*}
Therefore, we have
\begin{equation}\label{eq7}
E[\,\mbox{the number of sets in $\cup_{r\in R_\on} \F_r$ of type
(b)}~|~\mathcal E]=
\begin{cases}
|\F_R|+\Delta_2\cdot |L|^2+O(\Delta_2)\cdot |L| &\mbox{if}~F_{\max}=O(1)\\
|\F_R|+O(|L|) &\mbox{if}~F_{\max}=2
\end{cases}.
\end{equation}

Thus, we only focus on set $\sigma$ of type (a). Let
$\F_{\ell_\sigma,R_\off}$ be the sets containing $\ell_\sigma$ and
all of their elements in $R$ being \off\ at the end of Phase 1. Our
key observation here is that set $\sigma$ of type (a) will only
possibly become uncovered when an $r\in\sigma\cap R_\on$ turns \off\
if all but at most $\lceil c_{\max}/w_{\min}\rceil$ sets in
$\F_{\ell_\sigma,R_\off}$ have an element in $R$ that updates before
$r$ updates. Otherwise, $\ell_\sigma$ have too many uncovered sets
to turn \off\ before $r$ turns \off\ (hence, it remains \on). For an
arbitrary updating ordering of agents in $R\setminus\{r\}$, there
are at least $|\F_{\ell_\sigma,R_\off}|/\Delta_2$ elements that are
the first updating
agent in some set $\rho\in \F_{\ell_\sigma,R_\off}$. 
Therefore, for each $r\in \sigma\cap R_\on$,
$$\Pr[\sigma\in \F_r~|~\mathcal E]\leq\frac{\ceil{ c_{\max}/w_{\min}}+1}{|\F_{\ell_\sigma,R_\off}|/\Delta_2+1}.$$
Using the union bound,
$$\Pr[\sigma\in \cup_{r\in R_\on} \F_r~|~\mathcal E]\leq F_{\max}\cdot \frac{\ceil{ c_{\max}/w_{\min}}+1}{|\F_{\ell_\sigma,R_\off}|/\Delta_2+1}.$$
Now let $\F_{\ell,R}\subseteq\F$ be the sets containing $\ell\in L$
and at least one element of $R$.  Note also that given $\mathcal E$,
random variable $D_\ell:=|\F_{\ell,R_\off}|$ has (first-order)
dominance over the binomial random variable $X \sim
B\left(\frac{|\F_{\ell,R}|}{(F_{\max}-1)\Delta_2},\beta^{F_{\max}}\right)$.
Using this, we have
\begin{align}
E[\,\mbox{the number of sets in $\cup_{r\in R_\on} \F_r$ of type
(a)}~|~\mathcal E]&= \sum_{\ell\in
L}\sum_{\sigma\in\F_{\ell,R}}\E\left[\frac{\lceil
c_{\max}/w_{\min}\rceil +1}{D_\ell/\Delta_2+1}\right]\notag\\
&\leq \sum_{\ell\in L}\sum_{\sigma\in\F_{\ell,R}}O\left(\Delta_2\cdot\ceil{c_{\max}/w_{\min}}\right)\cdot \E\left[\frac{1}{D_\ell+1}\right]\notag\\
&\leq \sum_{\ell\in L}\sum_{\sigma\in\F_{\ell,R}}O\left(\Delta_2\cdot\ceil{c_{\max}/w_{\min}}\right)\cdot O\left(\frac{(F_{\max}-1)\Delta_2}{|\F_{\ell,R}|}\right)\notag\\
&=O\left(\Delta_2^2\cdot \ceil{c_{\max}/w_{\min}} \cdot|L|\right)\notag\\
&=O\left(\Delta_2^2\cdot|L|\right),\label{eq8}
\end{align}
where the second inequality uses the fact that
$\E[1/(1+Y)]\leq\frac1{np}$ for binomial random variable $Y\sim
B(n,p)$.

Finally, combining \eqref{eq5}, \eqref{eq6}, \eqref{eq7} and
\eqref{eq8} leads to the desired conclusion of Lemma \ref{LDstay},
where we remind that
$\texttt{cost}(s^{ad})=\Omega(|L|)+\Omega(|\F_R|)$ and $\Delta_2=1$
when $F_{\max}=2$.
\end{proof}

\section{Extension to Packing Games}\label{packing}

Notice that our covering games correspond to packing games if we simply redefine the
costs such that $i$ pays $c_i$ if he is \off\ and he pays the sum of the weights of
{\em fully-covered} sets he participates in if he is \on.  Roughly speaking, the game
strives to find a large packing, which is determined by the set of \on\ agents, while
avoiding fully covered sets.  Since we are simply relabeling actions, all the results from the previous sections apply.  The packing interpretation of this problem is easiest seen with the simple example where sets are of size 2, $c_i=c<1$ for all $i$, and $w_\sigma=1$ for all $\sigma$. In the original formulation of the problem, the sets of \on\ agents in Nash equilibria are minimal vertex covers.  In the new formulation, the sets of \on\ agents in Nash equilibria are maximal independent sets.  

\section{Conclusions}\label{conclusions} 
In recent years, game theoretic frameworks have provided informative
models for analyzing the outcomes of games among autonomous agents
or components programmed as autonomous agents. However, many games,
including those studied in this paper, often suffer from high Price
of Anarchy, meaning that without a central authority it is hard to
induce a state with low social cost. In this paper we study how weak
broadcasting signals from a central authority are enough to induce
states with low social cost in a general class of covering and
packing problems. In particular, we show that for any advertising
strategy $s^{ad}$, games with constantly bounded costs and weights
converge either in the public service advertising model of
\cite{BBM} or in the learn-then-decide model of \cite{bbm-ics-10} to
a state with cost $O(\texttt{cost}(s^{ad})^2)$. Moreover, in both
models we show convergence to a state of cost
$O(\texttt{cost}(s^{ad}))$ if all sets are of size 2. Furthermore,
for particular and poly-time computable $s^{ad}$ in the PSA model,
we guarantee convergence to a state within a $O(\log n)$ factor of
optimum for any game with sets of constant size.  We believe that
the techniques introduced in this paper to analyze covering and
packing games could be of broader interest for analyzing classic
optimization problems in a distributed fashion.

\subsection*{Acknowledgements}
This work was supported in part by ONR grant {N00014-09-1-0751}, by
AFOSR grant {FA9550-09-1-0538}, by NSF Career grant {CCF-0953192}.

\bibliographystyle{abbrv}
\bibliography{snbib1}

\appendix

\section{Proof of Proposition \ref{tech}}\label{sec:pftech}

The desired conclusion of Proposition \ref{tech} for the case $c<1$
follows since $d(1-a)^d=O(1)$ for all $d\geq 0$ as long as
$a\in(0,1)$ is constant. Hence, assume $c\geq1$. Let ${\bar
a}=\max(a,1-a)$ and define $\xi\in (0,1)$ to be the largest constant
satisfying
\[
(e/\xi)^\xi<\sqrt{1/{\bar a}}
\]
For each $\ell$, we have either $d\leq c/\xi$ or $d>c/\xi$. For the
case with $d\leq c/\xi$, observe that with $c\leq d$, the desired
expression is at most
\[
d\sum_{i=0}^{d} \binom{d}{i}(1-a)^{d-i}a^i=d=O(c)
\]
Now consider when $d>c/\xi$. Observe that
\begin{align*}
d\sum_{i=0}^{\floor{c}} \binom{d}{i}(1-a)^{d-i}a^i~\leq~ d \cdot
{\bar a}^{d}\sum_{i=0}^{\floor{c}} \binom{d}{i}  ~\leq~ d\cdot {\bar
a}^{d}\sum_{i=0}^{\floor{c}} \frac{d^i}{i!}
\end{align*}
Further, we have
\begin{align*}
d\cdot {\bar a}^{d}\sum_{i=0}^{\floor{c}} \frac{d^i}{i!} &=O(1)\cdot {\bar a}^{d/2} \sum_{i=0}^{\floor{c}} \frac{d^i}{i!}\\
&= O(c)\cdot {\bar a}^{d/2}\cdot \frac{d^{\floor{c}}}{\floor{c}!} \\
&= O(c)\cdot {\bar a}^{d/2} \Big(\frac{d\cdot e}{\floor{c}}\Big)^{\floor{c}} \\
&= O(c)\cdot {\bar a}^{d/2} \Big(\frac{d\cdot e}{\xi\cdot d}\Big)^{\xi\cdot d} \\
&= O(c)\cdot {\bar a}^{d/2} \cdot {\bar a}^{-d/2} \\
&= O(c),
\end{align*}
where we use (a) $d\cdot {\bar a}^{d/2}$ is $O(1)$, (b) $d^i/i!$ is
increasing with respect to $i$ for $i<c<d$, (c)
$x!=\Omega((x/e)^x)$, (d) $c<\xi\cdot d$ and (e) the definition of
$\xi$. This completes the proof of Proposition \ref{tech}.

\end{document}